\documentclass[preprint,3p,authoryear]{elsarticle}

\usepackage{graphicx,amssymb,amsmath,subfigure,lineno,natbib}
\usepackage{framed}
\usepackage{epsfig,mathptmx,times,color,flushend,gensymb,fancyhdr}
\usepackage{soul}
\definecolor{shadecolor}{rgb}{1,0.8,0.3}

\newdefinition{remark}{Remark}
\newtheorem{lemma}{Lemma}
\newtheorem{assumption}{Assumption}
\newtheorem{theorem}{Theorem }
\newproof{proof}{Proof}

\allowdisplaybreaks

\journal{Engineering Applications of Artificial Intelligence, vol. 62, pp. 276-285, 2017.}
\begin{document}
\begin{frontmatter}
\title{A Self-Learning Disturbance Observer for Nonlinear Systems in Feedback-Error Learning Scheme}
\author[UIUC]{Erkan Kayacan}
\author[ISU]{Joshua M. Peschel}
\author[UIUC]{Girish Chowdhary}
\address[UIUC]{Coordinated Science Laboratory, University of Illinois at Urbana-Champaign, \\ 1206 W Gregory Drive, Urbana, Illinois 61801, USA. Tel. +1 217 721 8278 \\
e-mail: \{ erkank, girishc \}@illinois.edu }
\address[ISU]{Department of Agricultural and Biosystems Engineering, Iowa State University, Ames, Iowa 50011, USA \\
e-mail: peschel@iastate.edu }

\begin{abstract}
This paper represents a novel online self-learning disturbance observer (SLDO) by benefiting from the combination of a type-2 neuro-fuzzy structure (T2NFS), feedback-error learning scheme and sliding mode control (SMC) theory. The SLDO is developed within a framework of feedback-error learning scheme in which a conventional estimation law and a T2NFS work in parallel. In this scheme, the latter
learns uncertainties and becomes the leading estimator whereas the former provides the learning error to the T2NFS for learning system dynamics. A learning algorithm established on SMC theory is derived for an interval type-2 fuzzy logic system. In addition to the stability of the learning algorithm, the stability of the SLDO and the stability of the overall system are proven in the presence of time-varying disturbances. Thanks to learning process by the T2NFS, the simulation results show that the SLDO is able to estimate time-varying disturbances precisely as distinct from the basic nonlinear disturbance observer (BNDO) so that the controller based on the SLDO ensures robust control performance for systems with time-varying uncertainties, and maintains nominal performance in the absence of uncertainties.
\end{abstract}
\end{frontmatter}

\section{Introduction}

One of the most essential requirements for controllers is to be insensitive to uncertainties. Many control methods have been proposed to handle different types of uncertainties, (e.g. $H_{\infty}$ control \citep{Glover1988,Gahinet1994}, sliding mode control (SMC) \citep{slotine1991,Hongyi2013,Xu2014}, adaptive control \citep{Weichao2013,Weichao2013ie,Wei2014,Jianyong2015}, etc...). In $H_{\infty}$ control, uncertainties must be bounded in $H_{\infty}$-norm. This implies that disturbances must disappear suddenly and completely. However, this is not a realistic assumption for real-time applications. In SMC theory, integral SMC controller is proposed in the presence of uncertainties. It is a well known fact that the integral action may cause unwanted effects such as, large settling time and overshoots. Moreover, adaptive control systems may not have an ability to control uncertain systems with highly changing parameters \citep{Chen2015123}. As an  alternative method, disturbance observers (DOs) have been proposed since they are very crucial for control of systems due to the fact that uncertainties extensively exist in practice and are extremely difficult to be modeled. These uncertainties, such as parameter variations, noise, unmodeled dynamics and interactions between subsystems, must be taken for the controller design into account to have a capability of getting robustness. For this purpose, different DOs have been designed in literature to obtain robust control performance for systems \citep{Chen2000ie,Chen2000,Chen2016}.

In DO based control approaches, the model uncertainties and external disturbances are merged into one term and a control law contains the estimated disturbance value by a DO. The aims are to achieve performance specifications while stabilizing the system considering the nominal model of the system and remove the disturbance effect on the system \citep{Chen2004,Yang2013,Ginoya2014}. A nonlinear dynamics inversion control method was designed for the longitudinal autopilot of a missile in \citep{chen2003}. It was reported that the control method exhibit poor performance in case of unknown uncertainties while a basic nonlinear DO (BNDO) based nonlinear dynamics inversion control approach ensured robust performance against uncertainty. The same BNDO has been used to design a robust SMC controller for systems with mismatched uncertainties \citep{Yang2013}. The main drawback in these studies, the BNDO is only be able to estimate time-invariant disturbances. If disturbances are time-varying, then the BNDO gives bias estimates. Furthermore, there are well known nonlinear observers, such as extended Kalman filter (EKF), particle filtering and nonlinear moving horizon estimation methods. EKF works well if linear approximation is valid and noise on measurements is small \citep{Haseltine2005}. Besides, particle filter and nonlinear moving horizon estimation methods require very large computation time \citep{Daum2005}. For this reason, an observer, computationally cheap, is required to be able to estimate time-varying disturbances.

Type-2 fuzzy logic systems (T2FLSs) are proposed as the extended versions of type-1 fuzzy logic systems (T1FLSs) in literature \citep{5584629,Castillo201219,Shing2015}. It allows us to have more degrees of freedom for design than their type-1 counterparts so that it results in better capability of handling uncertainty \citep{Castro2011, 6706738, 6502676,Solis2015,7083756}. Type-2 fuzzy sets are especially preferred as the decision of the position of the membership functions (MFs) precisely is a very troublesome task \citep{885114, castillo2013universal}. However, the computational complexity of generalized type-2 fuzzy sets is very high. For this reason, the interval type-2 fuzzy sets were proposed to decrease computation time and made it feasible in real-time applications \citep{873577, Maldonado2013496, 6636071,6506097,Castillo2014, 7083756, Wagner2015}. 

Neuro-fuzzy structures as a model-free method have been widely used for control and identification of systems in literature. It is well know that the stability of systems controlled by model-free controllers cannot be proven. Therefore, the feedback-error learning scheme has been proposed for neuro-fuzzy structures to guarantee the global asymptotic stability of the system in a compact space \citep{Efe2000}. SMC achieves robustness to parametric uncertainty and high-frequency unmodeled dynamics; therefore, the SMC theory-based learning algorithms for neuro-fuzzy structures have been proposed to ensure the robustness of the overall system \citep{Kaynak2001, Kaynak2009}. Moreover, they ensure faster convergence rate than the traditional learning methods, such as gradient descent, Levenberg-Marquardt and particle swarm optimization, because they are computationally simple. There are numerous examples for SMC theory-based learning algorithms for artificial neural networks, type-1 and type-2 neuro-fuzzy structures \citep{Topalov2009,Kayacan2012}. 

The main contribution of this paper is to develop a novel online SLDO, which can learn the disturbance behavior of systems in time-varying case as distinct from basic nonlinear DOs (BNDOs), be solved in the range of millisecond, and robust against uncertainties. For this purpose, the T2NFS in feedback-error learning scheme is proposed due to fact that they are very suitable techniques for adaptive learning. Additionally, computationally efficient sliding mode learning algorithm is used as the training algorithm of the T2NFS because it is a powerful approach for the stability issue. Consequently, the use of the combination of T2NFS, feedback-error learning scheme and sliding mode control theory harmoniously allow to better handle uncertainties.

The major contributions of this paper are as follows: 
\begin{enumerate}
\item The first major contribution is that a novel estimation approach in the feedback-error learning scheme is developed for disturbance observer design for the first-time. 
\item The second major contribution is that the stability of the training algorithm has been always proven for feedback-error learning methods in literature. In this paper, in addition to the proof of the training algorithm, the overall system stability is proven considering the dynamics of the proposed SLDO by adding a robust term. To the best knowledge of the author, this is also the first-time such a stability analysis is ever proven. 
\end{enumerate}

The minor contributions of this paper are as follows: 
\begin{enumerate}
\item The first minor contribution is that the developed SLDO is solved within milliseconds; therefore, the required computation time for the SLDO is significantly less than other methods, such as particle filter and nonlinear moving horizon estimation methods. 
\item The second minor contribution is that the learning rate of SMC theory-based learning algorithm for the T2NFS is adaptive so that it is possible to estimate the disturbance without the knowledge about the upper bound of the disturbance and its derivatives.
\end{enumerate}

The paper consists of six sections: The formulation of the BNDO is given in Section \ref{sec_prob_form}. The SLDO benefiting from the T2NFS is developed in Section \ref{sec_SLDO}. The T2NFS and online learning algorithm established on SMC theory are respectively represented in Sections \ref{sec_neuro-fuzzy} and \ref{sec_SMClearning}. The stability of the SLDO is proven in Section \ref{sec_stability_SLDO}. The controller design and the stability of the system are given in Section \ref{sec_controller}. The simulation results are represented in Sections \ref{sec_simulation}. Finally, the paper is summarized  in Section \ref{sec_conc}.

\section{Basic Nonlinear Disturbance Observer}\label{sec_prob_form}

A number of physical systems, such as robots, spacecrafts and mechanical systems, are generally described by second-order differential equations. A second-order nonlinear system is written in the following form:
\begin{equation}\label{eq_nonlinearsystem}
\dot{x } =  g_{1}(x) + g_{2}(x) u + z d(t)
\end{equation}
where $x=[x_{1},x_{2}]^{T}$ is the state vector, $u$ is the control input, $d(t)$ is the disturbance, $z=[z_{1},z_{2}]^{T}$ is the disturbance coefficients vector, $g_{1}(x)=[x_{2}, a(x)]^{T}$ and $g_{2}(x)=[0, b(x)]^{T}$ are the nonlinear system dynamics.

The disturbance $d(t)$ in \eqref{eq_nonlinearsystem} is not measurable in practice. Therefore, it is required to be estimated in practice to obtain robust control performance of systems. The following basic nonlinear disturbance observer (BNDO) dynamics have been proposed in  \citep{chen2003,Yang2013} as follows:
 
\begin{eqnarray}\label{eq_DO_dynamics}
\dot{p} & = & - l_{p} z p - l_{p} \Big(  z l_{p} x + g_{1}(x) + g_{2}(x) u \Big) \nonumber \\
\hat{d}_{BN}& = & p + l_{p} x
\end{eqnarray}
where $p$, $l_{p}$ and $\hat{d}_{BN}$ denote respectively the internal state, proportional observer gain and estimated disturbance. By taking time derivative of the estimated disturbance considering  \eqref{eq_DO_dynamics}, the time derivative of the estimated disturbance $\dot{\hat{d}}_{BN}$ is obtained as:
\begin{equation}\label{eq_obs_error_3}
\dot{\hat{d}}_{BN} =  l_{p} z  e_{d} 
\end{equation}

If the time derivative of the actual disturbance $\dot{d}(t)$ is added into \eqref{eq_obs_error_3}, the error dynamics of the BNDO are obtained as follows:
\begin{eqnarray}\label{eq_obs_error_4}
\dot{d}(t) - \dot{\hat{d}}_{BN} & =  & -l_{p} z   e_{d}  +\dot{d}(t) \nonumber \\
\dot{e}_{d} & = &  - l_{p} z   e_{d}  +\dot{d}(t)
\end{eqnarray}
where $e_{d}= d(t) - \hat{d}_{BN} $ is the disturbance error.

\begin{assumption}\label{assumption_BNDO}
The time derivative of the actual disturbance is bounded and $\displaystyle \lim_{ t \to \infty} \dot{d} (t)= 0$.
\end{assumption}

If Assumption 1 is satisfied, then \eqref{eq_obs_error_4} is obtained as follows:
\begin{equation}\label{eq_obs_error_5}
\dot{e}_{d} + l_{p} z  e_{d}  = 0
\end{equation}

\begin{lemma}\label{lemma_BNDO}
\citep{chen2003}: If $l_{p}  z$ is positive, i.e. $l_{p} z>0$, the disturbance error dynamics in \eqref{eq_obs_error_5} converge to zero asymptotically. 
\end{lemma}

Lemma 1 implies that the estimated disturbance by the BNDO is able to track the actual disturbance of the system in \eqref{eq_nonlinearsystem} asymptotically in case Assumption 1 is satisfied.

\begin{remark}\label{remark_BNDO}
If the time derivative of the actual disturbance $\dot{d}(t)$ is not equal to zero, the error dynamics of th BNDO cannot converge to zero so that BNDO gives bias. Therefore, there exists always difference between the estimated and true values of the disturbance. Similar observers have been designed in literature and the same drawback has been reported in \citep{chen2003,Yang2013}.
\end{remark} 

\section{Self-learning Disturbance Observer}\label{sec_SLDO}

\begin{figure}[b!]
  \centering
  \includegraphics[width=3.6in]{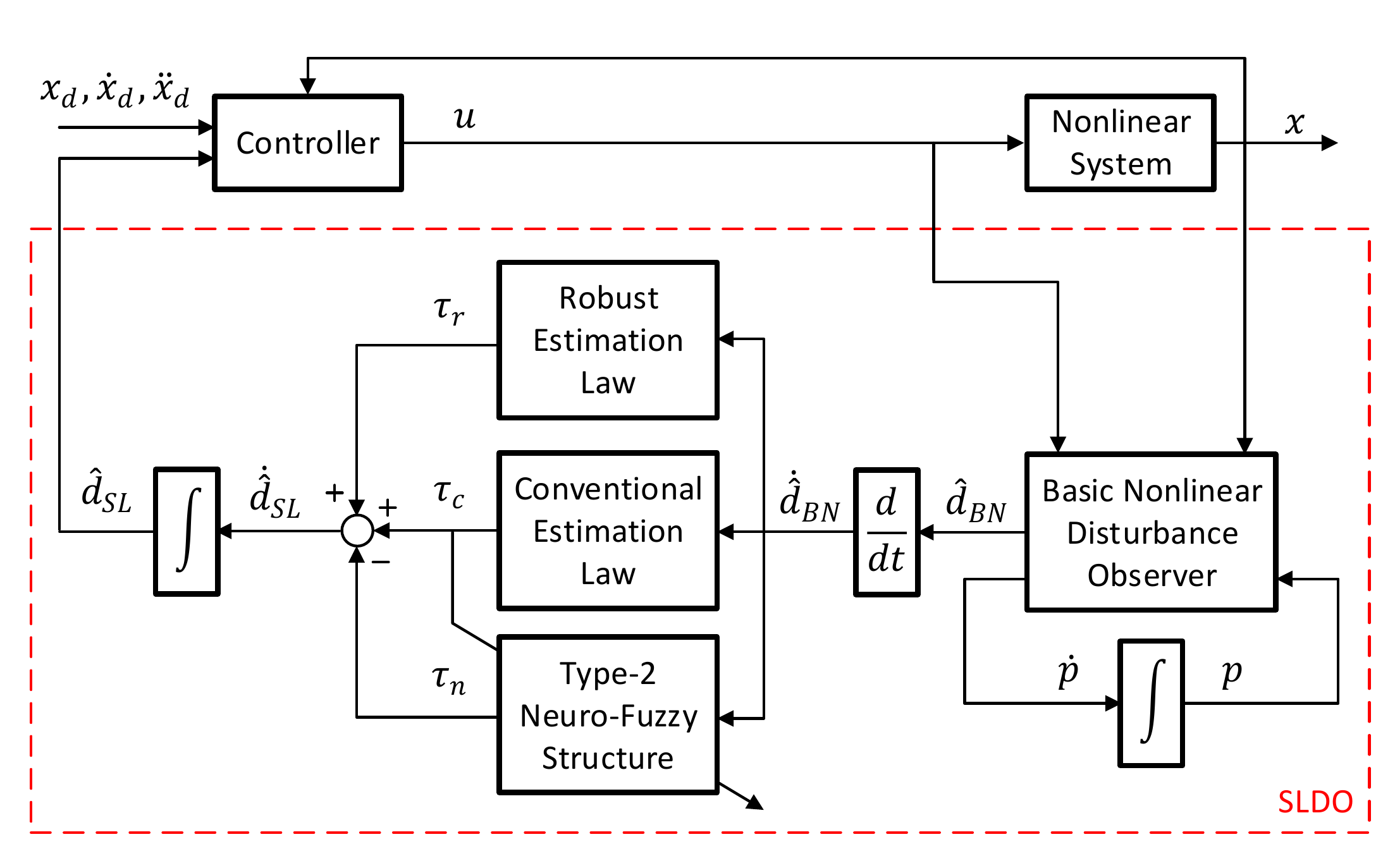}\\
  \caption{Schematic diagram of the self-learning disturbance observer (SLDO) }\label{fig_sldo_diagram}
\end{figure}

BNDOs cannot give unbiased estimation results in case of time-varying disturbances; therefore, a fast, computationally efficient, adaptive and robust disturbance observer is required. A novel estimation law is proposed as follows: 
\begin{equation}\label{eq_obs_estimation_law_self}
\dot{\hat{d}}_{SL} = \tau_{c} + \tau_{r} - \tau_{n}
\end{equation}
where $\tau_{c}$, $\tau_{r}$ and $\tau_{n}$ denote respectively the outputs of the conventional estimation law, robust term and T2NFS. The schematic diagram of the SLDO is illustrated in Fig \ref{fig_sldo_diagram}. As seen, the BNDO is working in series with the feedback-error learning algorithm in which the conventional and robust estimation laws work in parallel with T2NFS. 

The conventional  estimation law used in this paper is defined as follows:
\begin{equation}\label{eq_obs_estimation_law_conven}
\tau_{c} =  \dot{\hat{d}}_{BN} + \frac{l_{d} z}{l_{p} z}\ddot{\hat{d}}_{BN}
\end{equation}
where $\dot{\hat{d}}_{BN}$ denotes the time derivative of the estimated value of the disturbance by the BNDO, $l_{p}$ and $l_{d}$ denote the proportional and derivative gain vectors, and $l_{p} z$ and $l_{d} z$ are positive, i.e. $l_{p} z, l_{d} z > 0$.

The robust estimation law is written as follows:
\begin{equation}\label{eq_obs_estimation_law_self_robust}
\tau_{r} =  \frac{l_{r} z}{l_{p} z} \dot{\hat{d}}_{BN}
\end{equation}
$l_{r} z$ is positive, i.e. $l_{r} z >0$, and $l_{r}$ denotes the robust gain vector.

\subsection{Type-2 Neuro-Fuzzy Structure}\label{sec_neuro-fuzzy}

An interval type-2 Takagi-Sugeno-Kang (TSK) fuzzy \emph{if-then} rule $R_{ij}$ is written as:
\begin{equation}\label{eq_Rlinearfunction}
R_{ij}: \;\; \textrm{If} \; \xi_1 \; \textrm{is} \;\; \widetilde{1}_i \;\; \textrm{and} \; \xi_2 \; \textrm{is} \;\; \widetilde{2}_j, \;\; \textrm{then} \; f_{ij}=\Upsilon_{ij}
\end{equation}
where $\xi_1=\dot{\hat{d}}_{BN}$ and $\xi_2= \ddot{\hat{d}}_{BN}$ denote the inputs while $\widetilde{1}_i$ and $\widetilde{2}_j$  denote type-2 fuzzy sets for inputs. The function $f_{ij}$ is the output of the rules and the total number of the rules are equal to $K=I \times J$ in which $I$ and $J$ are the total number of the  inputs.

The upper and lower Gaussian membership functions for type-2 fuzzy logic systems are written as follows:
\begin{equation}\label{eq_mu1_lower}
\underline{\mu}_{1i}(\xi_1) = \exp\Bigg(-\bigg(\frac{\xi_{1}-\underline{c}_{1i}}{\underline{\sigma}_{1i}}\bigg)^2\Bigg)
\end{equation}
\begin{equation}\label{eq_mu1_upper}
\overline{\mu}_{1i}(\xi_1) = \exp\Bigg(-\bigg(\frac{\xi_{1}-\overline{c}_{1i}}{\overline{\sigma}_{1i}}\bigg)^2\Bigg)
\end{equation}
\begin{equation}\label{eq_mu2_lower}
\underline{\mu}_{2j}(\xi_2) = \exp\Bigg(-\bigg(\frac{\xi_{2}-\underline{c}_{2j}}{\underline{\sigma}_{2j}}\bigg)^2\Bigg)
\end{equation}
\begin{equation}\label{eq_mu2_upper}
\overline{\mu}_{2j}(\xi_2) = \exp\Bigg(-\bigg(\frac{\xi_{2}-\overline{c}_{2j}}{\overline{\sigma}_{2j}}\bigg)^2\Bigg)
\end{equation}
where $\underline{c}, \overline{c}, \underline{\sigma}, \overline{\sigma}$ denote respectively the lower and upper mean, and the lower and upper standard deviation of the membership functions. These parameters are adjustable for the T2NFS.

The lower and upper membership functions  $\underline{\mu }$ and $\overline{\mu }$ of A2-C0 fuzzy system employed in this paper are determined for every signal. Then, the firing strength of rules are calculated as follows:
\begin{equation}\label{eq_wij_lower_upper}
\underline{w}_{ij} = \underline{\mu}_{1i}(\xi_1)  \underline{\mu}_{2j}(\xi_2) \quad \textrm{and} \quad
\overline{w}_{ij} = \overline{\mu}_{1i}(\xi_1) \overline{\mu}_{2j}(\xi_2)
\end{equation}

The output of the every fuzzy rule is a linear function $f_{ij}$ formulated in \eqref{eq_Rlinearfunction}.  The output of the network is formulated below:
\begin{equation}\label{eq_taun}
\tau_n=q \sum_{i=1}^{I}\sum_{j=1}^{J}f_{ij}\widetilde{\underline{w}}_{ij}+(1-q)\sum_{i=1}^{I}\sum_{j=1}^{J}f_{ij}\widetilde{\overline{w}}_{ij}
\end{equation}
where $\widetilde{\underline{{w}}}_{ij}$ and $\widetilde{\overline{{w}}}_{ij}$ are the normalized firing strengths of the lower and upper output signals of the neuron $ij$ are written as follows:
\begin{equation}\label{eq_wij_lower_upper_normalized}
\widetilde{\underline{w}}_{ij} = \frac{\underline{w}_{ij}}{\sum_{i=1}^{I}\sum_{j=1}^{J}\underline{w}_{ij}} \;\;\; \textrm{and} \;\;\;
\widetilde{\overline{w}}_{ij} = \frac{\overline{w}_{ij}}{\sum_{i=1}^{I}\sum_{j=1}^{J}\overline{w}_{ij}}
\end{equation}
The design parameter $q$ weights the participation of the lower and upper firing levels and is generally set to $0.5$. In this paper, it is formulated as a time-varying parameter in the next subsection. 

The vectors are defined as::
\begin{eqnarray}
\widetilde{\underline{W}}(t) &=& [\widetilde{\underline{w}}_{11}(t) \; \widetilde{\underline{w}}_{12}(t) \dots \widetilde{\underline{w}}_{21}(t) \dots \widetilde{\underline{w}}_{ij}(t) \dots \widetilde{\underline{w}}_{IJ}(t)]^{T} \nonumber \\
\widetilde{\overline{W}}(t) &=& [\widetilde{\overline{w}}_{11}(t) \; \widetilde{\overline{w}}_{12}(t) \dots \widetilde{\overline{w}}_{21}(t) \dots \widetilde{\overline{w}}_{ij}(t) \dots \widetilde{\overline{w}}_{IJ}(t)]^{T} \nonumber \\
F &=& [f_{11} \; f_{12} \dots f_{21} \dots f_{ij} \dots f_{IJ}] \nonumber
\end{eqnarray}
where these normalized firing strengths are between $0$ and $1$, i.e. $0<\widetilde{\underline{w}}_{ij} \leq 1$ and $0<\widetilde{\overline{w}}_{ij} \leq 1$. In addition, $\sum_{i=1}^{I}\sum_{j=1}^{J}\widetilde{\underline{w}}_{ij} = 1$ and $\sum _{i=1}^{I}\sum_{j=1}^{J}\widetilde{\overline{w}}_{ij} = 1$.

\subsection{SMC Theory-based Learning Algorithm}\label{sec_SMClearning}

The sliding surface $s$ is formulated as follows:
\begin{equation}\label{eq_slidingfsurface_general}
s \left(\tau_{c}\right)= \tau_{c}
\end{equation}
where $\tau_{c}$ is the output of the conventional estimation law and used as a sliding surface. It is to be noted that the sliding surface is used as learning error to train the SMC theory-based learning algorithm.

The adaptation rules of the T2NFS parameters are given by the following equations:
\begin{equation} \label{eq_c_1i_lower}
\dot{\underline{c}}_{1i} = \dot{\xi}_{1} + (\xi_{1} - \underline{c}_{1i}) \alpha \textrm{sgn}\left(s  \right)
\end{equation}
\begin{equation} \label{eq_c_1i_upper}
\dot{\overline{c}}_{1i} = \dot{\xi}_{1} + (\xi_{1} - \overline{c}_{1i}) \alpha \textrm{sgn}\left( s \right)
\end{equation}
\begin{equation} \label{eq_c_2j_lower}
\dot{\underline{c}}_{2j} = \dot{\xi}_{2} + (\xi_{2} - \underline{c}_{2j}) \alpha \textrm{sgn}\left( s \right)
\end{equation}
\begin{equation} \label{eq_c_2j_upper}
\dot{\overline{c}}_{2j} = \dot{\xi}_{2} + (\xi_{2} - \overline{c}_{2j}) \alpha \textrm{sgn}\left( s \right)
\end{equation}
\begin{equation}\label{eq_sigma_1i_lower}
\dot{\underline{\sigma}}_{1i} = - \bigg( \underline{\sigma}_{1i} + \frac{ (\underline{\sigma}_{1i} )^3}{(\xi_{1} - \underline{c}_{1i})^2} \bigg) \alpha  \textrm{sgn}\left( s \right)
\end{equation}
\begin{equation}\label{eq_sigma_1i_upper}
\dot{\overline{\sigma}}_{1i} = - \bigg( \overline{\sigma}_{1i} + \frac{ (\overline{\sigma}_{1i} )^3}{(\xi_{1} - \overline{c}_{1i})^2} \bigg) \alpha  \textrm{sgn}\left( s \right)
\end{equation}
\begin{equation}\label{eq_sigma_2j_lower}
\dot{\underline{\sigma}}_{2j} = - \bigg( \underline{\sigma}_{2j} + \frac{ (\underline{\sigma}_{2j} )^3}{(\xi_{2} - \underline{c}_{2j})^2} \bigg) \alpha  \textrm{sgn}\left( s  \right)
\end{equation}
\begin{equation}\label{eq_sigma_2j_upper}
\dot{\overline{\sigma}}_{2j} = - \bigg( \overline{\sigma}_{2j} + \frac{ (\overline{\sigma}_{2j} )^3}{(\xi_{2} - \overline{c}_{2j})^2} \bigg) \alpha  \textrm{sgn}\left( s  \right)
\end{equation}
\begin{equation}\label{eq_f_ij}
\dot{f}_{ij} =-\frac{q \widetilde{\underline{w}}_{ij}+ (1-q)\widetilde{\overline{w}}_{ij}}{(q\widetilde{\underline{W}}+(1-q) \widetilde{\overline{W}})^T(q\widetilde{\underline{W}}+ (1-q)\widetilde{\overline{W}})}\alpha sgn( s)
\end{equation}
\begin{equation}\label{eq_q}
\dot{q} =-\frac{1}{F(\widetilde{\underline{W}}-\widetilde{\overline{W}})^{T}}\alpha sgn( s  )
\end{equation}
\begin{equation}\label{eq_alpha}
\dot{\alpha} = \gamma_{\alpha} \mid s \mid 
\end{equation}
where $\alpha$ and $\gamma_{\alpha}$ denote respectively the learning rate and the coefficient of the adaptation for the learning rate, and they must be positive, i.e. $\alpha, \gamma_{\alpha} > 0$.

\begin{theorem}[Stability of the learning algorithm]\label{theorem1}
If adaptations rules are proposed as in \eqref{eq_c_1i_lower}-\eqref{eq_alpha} and the final value of the learning rate $\alpha^{*}$ is large enough, i.e. $\alpha^{*} > \dot{\tau}_{r}^{*} + \ddot{\hat{d}}^{*}_{SL} $ where $\dot{\tau}_{r}^{*} $ and $\ddot{\hat{d}}^{*}_{SL}$ are respectively the upper bounds of $\dot{\tau}_{r} $ and $\ddot{\hat{d}}_{SL}$, this ensures that $\tau_{c}$ will converge to zero in finite time for a given arbitrary initial condition $\tau_{c}(0)$.  
\end{theorem}

\begin{proof}
The Lyapunov function is written as follows:
\begin{eqnarray}\label{eq_smc_Lyapunov}
V = \frac{1}{2} \tau_{c}^{2} +  \frac{1}{\gamma_{\alpha} } (\alpha - \alpha^{*})^{2} 
\end{eqnarray}
By taking the time derivative of the Lyapunov function in \eqref{eq_smc_Lyapunov}, it is obtained as follows:
\begin{eqnarray}\label{eq_smc_Lyapunov_d} 
\dot{V} =   \tau_{c} \dot{\tau}_{c} +  \frac{2 \dot{\alpha} }{\gamma_{\alpha}} (\alpha - \alpha^{*})
\end{eqnarray}
If the \eqref{eq_alpha} is inserted into the equation above:
\begin{eqnarray}\label{eq_smc_Lyapunov_dd} 
\dot{V} =   \tau_{c} (\dot{\tau}_{n} - \dot{\tau}_{r}+ \ddot{\hat{d}}_{SL} ) +2 \mid s\mid (\alpha - \alpha^{*}) 
\end{eqnarray}
The calculation of $\dot{\tau}_{n}$ in \eqref{dotVc4} is inserted into \eqref{eq_smc_Lyapunov_dd}, it is obtained as follows:
\begin{eqnarray}\label{eq_smc_Lyapunov_d2}
\dot{V} =   \tau_{c} \Big( -2\alpha  \textrm{sgn}\left( s \right) - \dot{\tau}_{r} + \ddot{\hat{d}}_{SL} \Big) +2 \mid s \mid  (\alpha - \alpha^{*}) 
\end{eqnarray}
If \eqref{eq_slidingfsurface_general} is inserted into the equation above:
\begin{eqnarray}\label{eq_smc_Lyapunov_d2}
\dot{V} =   \tau_{c} \Big( -2\alpha  \textrm{sgn}\left( \tau_{c}\right) - \dot{\tau}_{r} + \ddot{\hat{d}}_{SL} \Big) +2 \mid \tau_{c} \mid (\alpha - \alpha^{*}) 
\end{eqnarray}
If it is assumed that $\ddot{\hat{d}}_{SL}$ and $\dot{\tau}_{r}$ are upper bounded by $\ddot{\hat{d}}^{*}_{SL}$ and $\dot{\tau}_{r}^{*}$, \eqref{eq_smc_Lyapunov_d2} is obtained as follows:
\begin{eqnarray}\label{eq_smc_Lyapunov_d3}
\dot{V} &= &  \mid \tau_{c} \mid ( -2 \alpha + \dot{\tau}_{r}^{*} + \ddot{\hat{d}}^{*}_{SL} ) +2 \mid \tau_{c}\mid (\alpha - \alpha^{*})  \nonumber \\
&= &  \mid \tau_{c} \mid ( -2 \alpha^{*} + \dot{\tau}_{r}^{*} + \ddot{\hat{d}}^{*}_{SL} ) 
\end{eqnarray}
As stated in Theorem \ref{theorem1}, if the final value of the learning rate $\alpha^{*}$ is large enough, i.e. $\alpha^{*} >  \dot{\tau}_{r}^{*} + \ddot{\hat{d}}^{*} _{SL}$, then the time derivative of the Lyapunov function is negative, i.e. $\dot{V}<0$ so that the SMC theory-based learning algorithm is stable and $\tau_{c}$ will converge to zero in finite time.
\end{proof}

\begin{remark}\label{remark_smctbla}
Since the adaptation rules in \eqref{eq_alpha} are
enforced, the final value of the learning rate of the T2NFS is determined during the adaptation of learning rate, and it is able to reach large values to make learning algorithm stable. This is a superiority of the proposed approach in this study as distinct from previous studies in which the upper bounds are needed to be foreknown.
\end{remark}

SMC theory endures high-frequency oscillations, i.e. chattering.  In this paper, the function in \eqref{eq_chatter} has been proposed to remove the chattering effect as the sign function in \eqref{eq_c_1i_lower}-\eqref{eq_alpha}.
\begin{equation}\label{eq_chatter}
\textrm{sgn}(s):=\frac{s}{ \mid s \mid + \delta}
\end{equation}
where $\delta=0.05$.

\begin{remark}\label{remark_smctbla2}
The usage of  the sliding surface $s$ in \eqref{eq_slidingfsurface_general} as learning error for the T2NFS with the adaptation laws in \eqref{eq_c_1i_lower}-\eqref{eq_alpha} accomplishes the desired sliding regime for the observer. 
\end{remark}

\subsection{Stability Analysis}\label{sec_stability_SLDO}

The proposed SLDO law in \eqref{eq_obs_estimation_law_self} is re-written considering \eqref{eq_obs_estimation_law_conven} and \eqref{eq_obs_estimation_law_self_robust} as:
\begin{eqnarray}\label{eq_obs_estimation_law_self_stability}
\dot{\hat{d}}_{SL} =  (1 + \frac{l_{r} z}{l_{p} z} )  \dot{\hat{d}}_{BN} +  \frac{l_{d} z}{l_{p} z} \ddot{\hat{d}}_{BN}  - \tau_{n}
\end{eqnarray}

The error dynamics for the SLDO are obtained by adding the actual disturbance rate $\dot{d}$ into the estimated disturbance rate in \eqref{eq_obs_estimation_law_self_stability} and considering the calculated time derivative of the estimated disturbance by BNDO in \eqref{eq_obs_error_3}:
\begin{eqnarray}\label{eq_obs_error_self_1}
 \dot{d} - \dot{\hat{d}}_{SL}  &=&  - ( l_{p} + l_{r} ) z e_{d} -  l_{d} z \dot{e}_{d} + \tau_{n} + \dot{d} \nonumber \\
\dot{e}_{d}  &=&  \frac{ - ( l_{p} + l_{r} ) z e_{d}  + \tau_{n} + \dot{d} } { 1+ l_{d} z }
\end{eqnarray}

By taking the time derivative of \eqref{eq_obs_error_self_1}, it is obtained as follows:
\begin{equation}\label{eq_obs_error_self_2}
\ddot{e}_{d}  =  \frac{ - ( l_{p} + l_{r} ) z \dot{e}_{d}  + \dot{\tau}_{n} + \ddot{d} } { 1+ l_{d} z }
\end{equation}
As calculated in  \eqref{dotVc4}, $\dot{\tau}_{n}=- 2 \alpha sgn(s)$ is inserted into \eqref{eq_obs_error_self_2};
\begin{equation}\label{eq_obs_error_self_3}
\ddot{e}_{d}  =  \frac{ -(  l_{p} + l_{r} ) z \dot{e}_{d} - 2 \alpha sgn(s) + \ddot{d}} {1 + l_{d} z}
\end{equation}

If  $\tau_{c}$ in \eqref{eq_obs_estimation_law_conven} is inserted into \eqref{eq_slidingfsurface_general}, the sliding surface is obtained as follows:
\begin{eqnarray}\label{eq_slidingfsurface}
s \left( \dot{\hat{d}}_{BN}, \ddot{\hat{d}}_{BN} \right)  =  \frac{l_d z}{l_p z} \Big(\ddot{\hat{d}}_{BN} + \lambda \dot{\hat{d}}_{BN} \Big ) 
\end{eqnarray}
where $\lambda=\frac{l_p z}{l_d z}$ is the slope of the sliding surface. 
The time derivative of the sliding surface is obtained as
\begin{equation}\label{eq_slidingfsurface_dot}
\dot{s} = \frac{l_{d} z}{l_{p} z}  \Big(\dddot{\hat{d}}_{BN} + \frac{l_{p} z}{l_{d} z}  \ddot{\hat{d}}_{BN} \Big ) 
\end{equation}

\begin{theorem}[Stability of the SLDO]\label{theorem2}
The estimation law in \eqref{eq_obs_estimation_law_self} is employed as a DO, the closed-loop error dynamics for the SLDO are stable if the robust gain $l_{r}$ is equal to $\frac{l_{p}}{l_{d} z}$, i.e. $l_{r} = \frac{l_{p}}{l_{d} z}$, and the  final value of the learning rate of T2NFS $\alpha^{*}$ is large enough, $\alpha^{*} > \ddot{d}^{*} $ where the acceleration of the actual disturbance $\ddot{d}$ is upper bounded by $\ddot{d}^{*}$.
\end{theorem}

\begin{proof}
The Lyapunov function is written as follows:
\begin{equation}\label{eq_observer_Lyapunov}
V= \frac{1}{2}  s^{2} + \frac{ l_{d} z }{\gamma_{\alpha} (1 + l_{d} z)} (\alpha - \alpha^{*})^{2}
\end{equation}
By taking the time derivative of the Lyapunov function above considering \eqref{eq_alpha}, it is obtained as
\begin{eqnarray}\label{eq_observer_Lyapunov_d} 
\dot{V} &=&  s \dot{s} +  l_{d} z \frac{2\mid s \mid }{1 + l_{d} z} (\alpha - \alpha^{*}) 
\end{eqnarray}
If the time derivative of the sliding surface is inserted into the aforementioned equation, it is obtained as follows:  
\begin{eqnarray}\label{eq_observer_Lyapunov_dd} 
 \dot{V} &=&  s \frac{l_{d} z}{l_{p} z}  \Big(\dddot{\hat{d}}_{BN} + \frac{l_{p} z}{l_{d} z}  \ddot{\hat{d}}_{BN} \Big )  +  l_{d} z  \frac{2\mid s \mid}{1 + l_{d} z} (\alpha - \alpha^{*}) 
\end{eqnarray}
It is obtained considering \eqref{eq_obs_error_3}
\begin{eqnarray}\label{eq_observer_Lyapunov_d1} 
 \dot{V} &=&  s \Big( l_{d} z \ddot{e}_{d}  + l_{p} z \dot{e}_{d}   \Big)  +  l_{d} z  \frac{2\mid s \mid}{1 + l_{d} z} (\alpha - \alpha^{*}) 
\end{eqnarray}
\eqref{eq_obs_error_self_3} is inserted into \eqref{eq_observer_Lyapunov_d1}, it is obtained as follows:
\begin{eqnarray}\label{eq_observer_Lyapunov_d2}
\dot{V} & = & s \Big(  l_{d} z \frac{ - (  l_{p} + l_{r} ) z \dot{e}_{d} - 2 \alpha sgn(s) + \ddot{d}} {1 + l_{d} z} + l_{p} z   \dot{e}_{d} \Big)  +   l_{d} z \frac{2\mid s \mid}{1 + l_{d} z} (\alpha - \alpha^{*})
\end{eqnarray}
If it is assumed that $\ddot{d}$ is upper bounded by $\ddot{d}^{*}$:
\begin{eqnarray}\label{eq_observer_Lyapunov_d3}
 \dot{V} &=&  \mid s \mid l_{d} z \frac{(- 2 \alpha + \ddot{d}^{*})}{1 + l_{d} z} + s \dot{e}_{d}  \Big(l_{d} z \frac{ -(  l_{p} + l_{r} ) z }{1 + l_d z}  + l_p z \Big) +  l_{d} z  \frac{2\mid s \mid}{1 + l_{d} z} (\alpha - \alpha^{*})  \\
 &=&  \mid s \mid  l_{d} z \frac{(- 2 \alpha^{*} + \ddot{d}^{*})}{1 + l_{d} z} + s  \dot{e}_{d}  \underbrace{ \Big( l_{d} z \frac{ -(  l_{p} + l_{r} ) z }{1 + l_d z} + l_pz \Big) }_{0}
\end{eqnarray}

As stated in Theorem \ref{theorem2}, if $l_{r}$ is equal to $\frac{l_{p}}{l_{d} z}$, i.e. $l_r = \frac{l_p}{l_d z}$, and the final value of the learning algorithm $\alpha^{*}$ is large enough, i.e. $\alpha^{*} >\ddot{d}^{*}$, then the time derivative of the Lyapunov function is negative, i.e. $\dot{V}<0$, so that the SLDO is stable.
\end{proof}

\begin{remark}\label{remark_stability}
The main advantage of the SLDO is to be able to prove the stability in case of not only time-invariant disturbances, such as BNDOs, but also time-varying disturbances.
\end{remark}

\section{Controller Design}\label{sec_controller}

The control objective is to find a control law so that the system states can track a desired trajectory. One of the most commonly used method for nonlinear systems is feedback linearization control (FLC). The traditional FLC method for nonlinear systems is formulated considering a second-order nonlinear system in \eqref{eq_nonlinearsystem} where there exists no disturbance:
\begin{eqnarray}\label{eq_FLC}
u= -b^{-1}(x) \Big( \ddot{x}_{d} + a (x) - k_{2} (\dot{x}_{d} - x_{2} ) - k_{1}(x_{d} - x_{1} ) \Big)
\end{eqnarray}
where the controller coefficients $k_{1}, k_{2}$ are positive, i.e. $k_{1}, k_{2}>0$.  If the control law in \eqref{eq_FLC} is applied to the system in \eqref{eq_nonlinearsystem}, the closed-loop error dynamics are obtained as follows:
\begin{eqnarray}\label{eq_FLC_closedloop_tra}
\ddot{e} + k_{2} \dot{e} + k_{1} e =  - z d(t)
\end{eqnarray}
where $e=x_{d}- x_{1}$ and $\dot{e}=\dot{x}_{d}- x_{2}$. 

\begin{lemma}\label{lemma_i2sstability}
\citep{khalil1996nonlinear}: If a nonlinear system $F(x,u)$ is input-to-state stable and the input satisfies $\displaystyle \lim_{ t \to \infty} u(t) = 0$, then the state satisfies $\displaystyle \lim_{ t \to \infty} x(t) = 0$.
\end{lemma}

\begin{remark}\label{remark_flctra}
As seen in \eqref{eq_FLC_closedloop_tra}, if the disturbance $d(t)$ is different from zero, the closed-loop error dynamics cannot converge to zero in finite time. This shows that the traditional FLC is sensitive to disturbances.
\end{remark}

The FLC based on the SLDO by taking the estimated disturbance value into account is formulated as follows:
\begin{eqnarray}\label{eq_FLC_new}
u= -b^{-1}(x) \Big(\ddot{x}_{d} + a (x) - k_{2} (\dot{x}_{d} - x_{2} ) - k_{1}(x_{d} - x_{1}  ) + z \hat{d}_{SL} \Big)
\end{eqnarray}
If the control law in \eqref{eq_FLC_new} is applied to the system in \eqref{eq_nonlinearsystem}, the closed-loop error dynamics are obtained as follows:
\begin{eqnarray}\label{eq_FLC_closedloop_new}
\ddot{e} + k_{2} \dot{e} + k_{1} e = - z e_{d} 
\end{eqnarray}
where  $e_{d}=d(t)  - \hat{d}_{SL}(t)$. As stated in Theorem \ref{theorem2}, the disturbance error dynamics for the SLDO can converge to zero asymptotically. As stated in Lemma \ref{lemma_i2sstability}, if the disturbance error $e_{d}$ satisfies $\displaystyle \lim_{ t \to \infty} e_{d}(t) = 0$, then the system error satisfies $\displaystyle \lim_{ t \to \infty} e(t) = 0$. As a result, the closed-loop error dynamics of the system can converge to zero asymptotically in finite  time under the control law in \eqref{eq_FLC_new} if the controller coefficients $k_{1}$ and $k_{2}$ are positive, i.e. $k_{1}, k_{2} > 0$.

\begin{remark}\label{remark_flcnew}
 If there exists no disturbance, i.e. $d(t)=0$, then the estimated value of the disturbance in the control law \eqref{eq_FLC_new} will be zero, i.e. $\hat{d}_{SL}=0$. This results in the traditional FLC in \eqref{eq_FLC} so that it maintains the nominal performance the absence of disturbances.
\end{remark}

\section{Simulation Studies}\label{sec_simulation}

The following nonlinear system, i.e., chaotic Duffing oscillator, is considered for the simulation studies \citep{Hsu2012997}:
\begin{eqnarray}\label{eq_systems}
\dot{x}_{1}& = & x_{2}  \nonumber \\
\dot{x}_{2} & = & 1.1 x_{1} - 0.4 x_{2} - x_{1}^{3} +2.1 \cos{(1.8 t)}+ u + d
\end{eqnarray}
where $a(x)= 1.1 x_{1} - 0.4 x_{2} - x_{1}^{3} +2.1 \cos{(1.8 t)}$,  $b(x)=1$ and $z=[0,1]^{T}$ as can be seen from \eqref{eq_nonlinearsystem}.

The desired states values are defined as $\ddot{x}_{d}=\dot{x}_{d}=x_{d}=0$. The initial conditions on the states of the system and the controller coefficients are respectively selected as $x(0)=[1, -1]^{T}$ and $k_{1}=50$,  $k_{2}=25$. Since the disturbance coefficient vector in \eqref{eq_systems} is equal to $z=[0, 1]^{T}$, the proportional, derivative and robust gains for DOs must be positive, i.e., $l_{p}, l_{d}, l_{r} >0$. The proportional gain is selected as $l_{p}=[0, 3]^{T}$ while the derivative gain is selected as $l_{d}=[0, 1.2]^{T}$. As stated in Theorem \ref{theorem2}, since the robust gain must be equal to $\frac{l_{p}}{l_{d} z}$, i.e. $l_{r} = \frac{l_{p}}{l_{d} z}= 2.5$, the robust gain is selected as $l_{r}=[0, 2.5]^{T}$. The coefficient $\gamma_{\alpha}$ to adjust the learning rate $\alpha$ for the SLDO is  selected as $0.001$. The initial conditions on the learning rate $\alpha$ and parameter $q$ are set to $0.05$ and $0.5$, respectively. To benchmark different disturbance observers in the presence and absence of uncertainties, no disturbance is imposed on the system at the beginning, a step external disturbance $d = 3$ is imposed on the system at $t = 10$ second and a sinusoidal external disturbance $d=3sin(t)$ is imposed on the system at $t=20$ second as formulated below:
\begin{equation}\label{eq_aplieddisturbance}
  d(t) = \Bigg \{
    \begin{array}{rclcl}
     0 \leq &t& < 10 &  d=& 0 \\
   10 \leq &t& < 20 & d=& 3 \\
   20 \leq &t& < 30 & d=& 3 \sin(t) \\
 \end{array} 
\end{equation}

In simulation studies, the control performance of the FLC based on the SLDO is firstly compared with the traditional FLC and the FLC based on the BNDO. Then, the SLDO is analyzed under noisy conditions and compared with its type-1 counterpart. Throughout simulation studies, the sampling time is set to $0.001$ second while the number of membership functions are selected as $I = J = 3$. In the presence of plant uncertainties, the adaptation of the learning rate of sliding mode learning algorithm must be a robust adaptation law to avoid having infinite values. Therefore, a dead-zone has been proposed in literature to handle this problem. In this paper, if the sliding surface is smaller than the dead-zone parameter $\epsilon=0.05$, i.e., $s < \epsilon$, then the learning rate $\alpha$ is not updated.

The states responses $x_{1}, x_{2}$ are shown in Figs. \ref{fig_x1} and \ref{fig_x2}. Firstly, the FLC controller can control the system without steady-state error while there exists no disturbance on the system. However, after the disturbances are imposed on the system, it is observed that it is not robust against any external disturbance and gives steady-state error as seen in Fig. \ref{fig_x1} and stated in Remark \ref{remark_flctra}. Secondly, the FLC based on the BNDO can control the system without any steady-state error while there exist no disturbance and a time-invariant disturbance. However, it is seen that it is not robust against a time-varying disturbance as stated in Remark \ref{remark_BNDO}. Thirdly, the FLC based on the SLDO can control system without steady-state error and it is robust against time-invariant and time-varying disturbances. Moreover, the FLCs based on the BNDO and SLDO maintain the nominal performance while there exists no disturbance between t=0-10 seconds as stated in Remark \ref{remark_flcnew}.

The actual and estimated disturbances are shown in Fig. \ref{fig_d}. As can be observed, the SLDO can estimate time-varying disturbances while BNDO is only able to estimate only time-invariant disturbances as stated in Remark \ref{remark_stability}. This fact results in the robust control performance of the FLC based on the SLDO against time-varying uncertainties. Thanks to learning process by the feedback-error learning algorithm, the T2NFS takes the overall estimation signal while the conventional estimation signal converges to zero in finite time as shown in Fig. \ref{fig_estimation_signals}. Inasmuch as the total generated estimation signal by the feedback-error learning structure is equal to $\tau_{c}-\tau_{n}$, the output of the T2NFS $\tau_{n}$ is multiplied by $-1$ in Fig. \ref{fig_estimation_signals} not to cause the reader to become perplexed. The T2NFS becomes the leading estimator after a short time period. The output of the conventional estimation law $\tau_{c}$ becomes nonzero only during the time intervals when the T2NFS is learning.

The adaptation of the learning rate $\alpha$ is shown in Fig \ref{fig_alpha}. As seen, the initial condition on the learning rate $\alpha$ is set to $0.05$ and the learning rate is constant while the output of the conventional estimation law $\tau_{c}$ is equal to zero due to the fact that learning is not required. When the disturbances are imposed on the system, the learning rate is increasing for a short time period till $\tau_{c}$ becomes zero. Moreover, the adaptation of the parameter $q$ is shown in Fig. \ref{fig_q}. The initial condition on the parameter $q$ is set to $0.5$, which is the general case. Thanks to the adaptation rule in \eqref{eq_q}, the proportion of the upper and lower membership functions is adjusted throughout the simulations. 

\begin{figure}[h!]
\centering
\subfigure[ ]{
\includegraphics[width=3in]{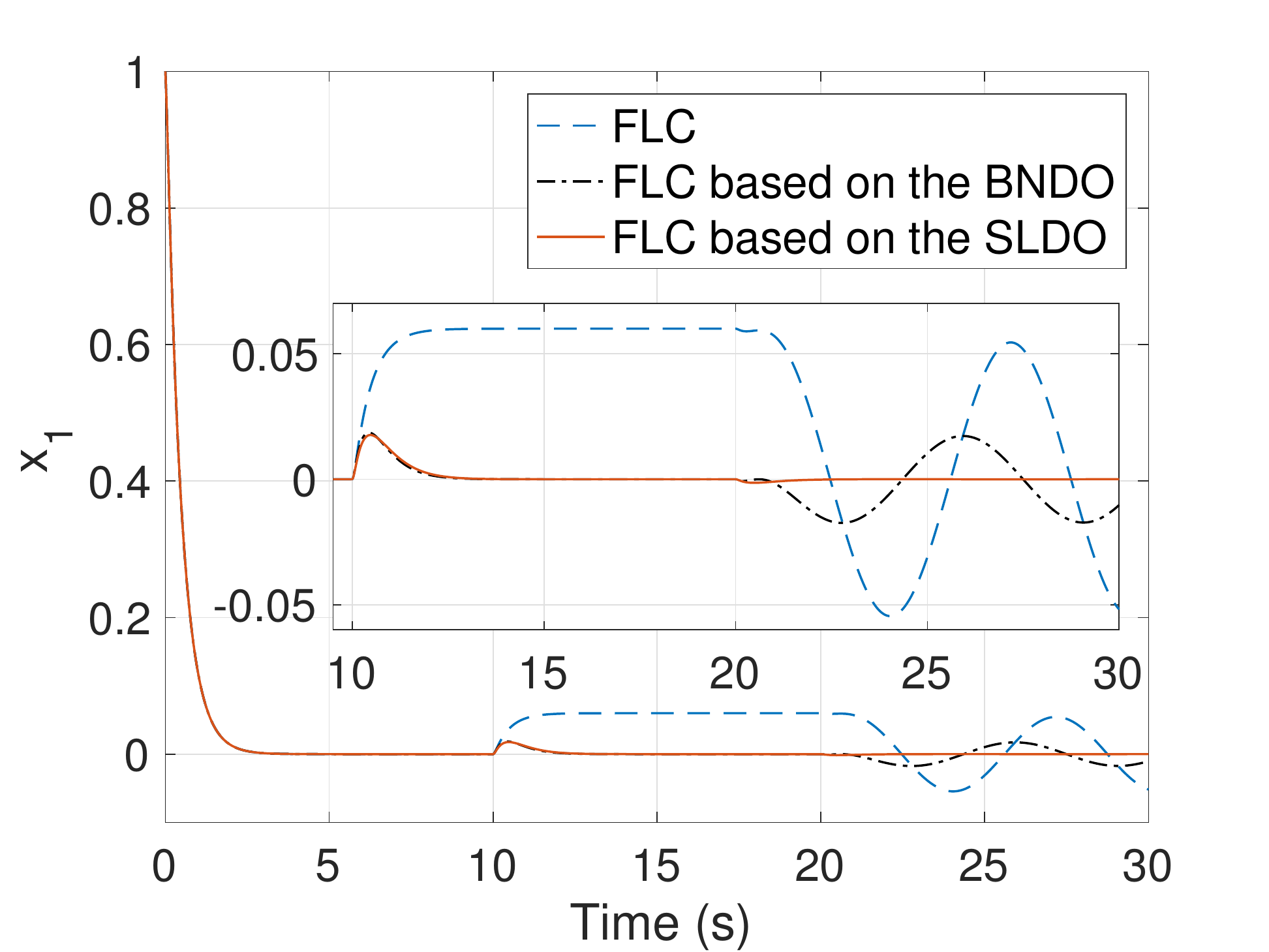}
\label{fig_x1}
}
\subfigure[ ]{
\includegraphics[width=3in]{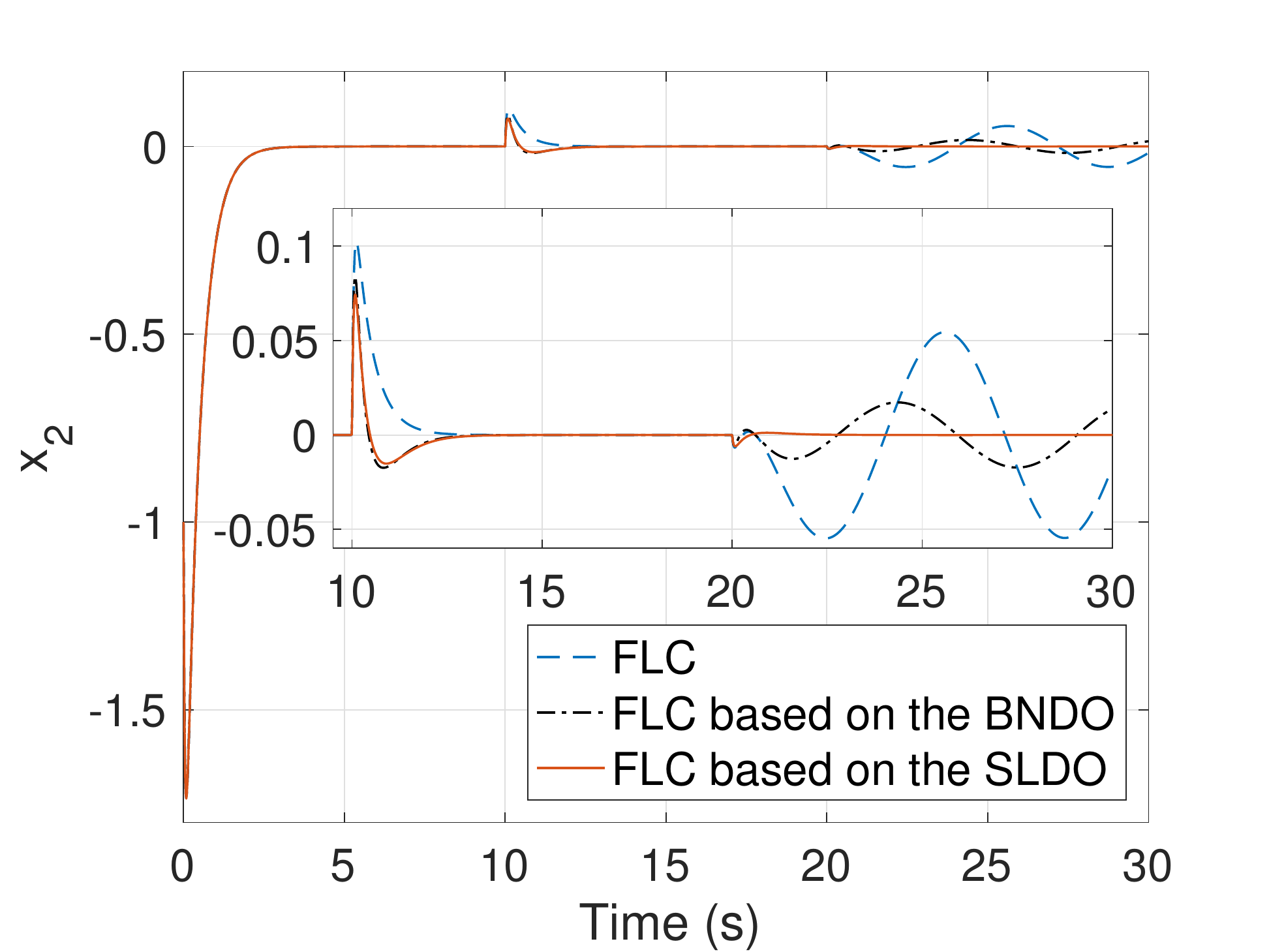}
\label{fig_x2}
}
\subfigure[ ]{
\includegraphics[width=3in]{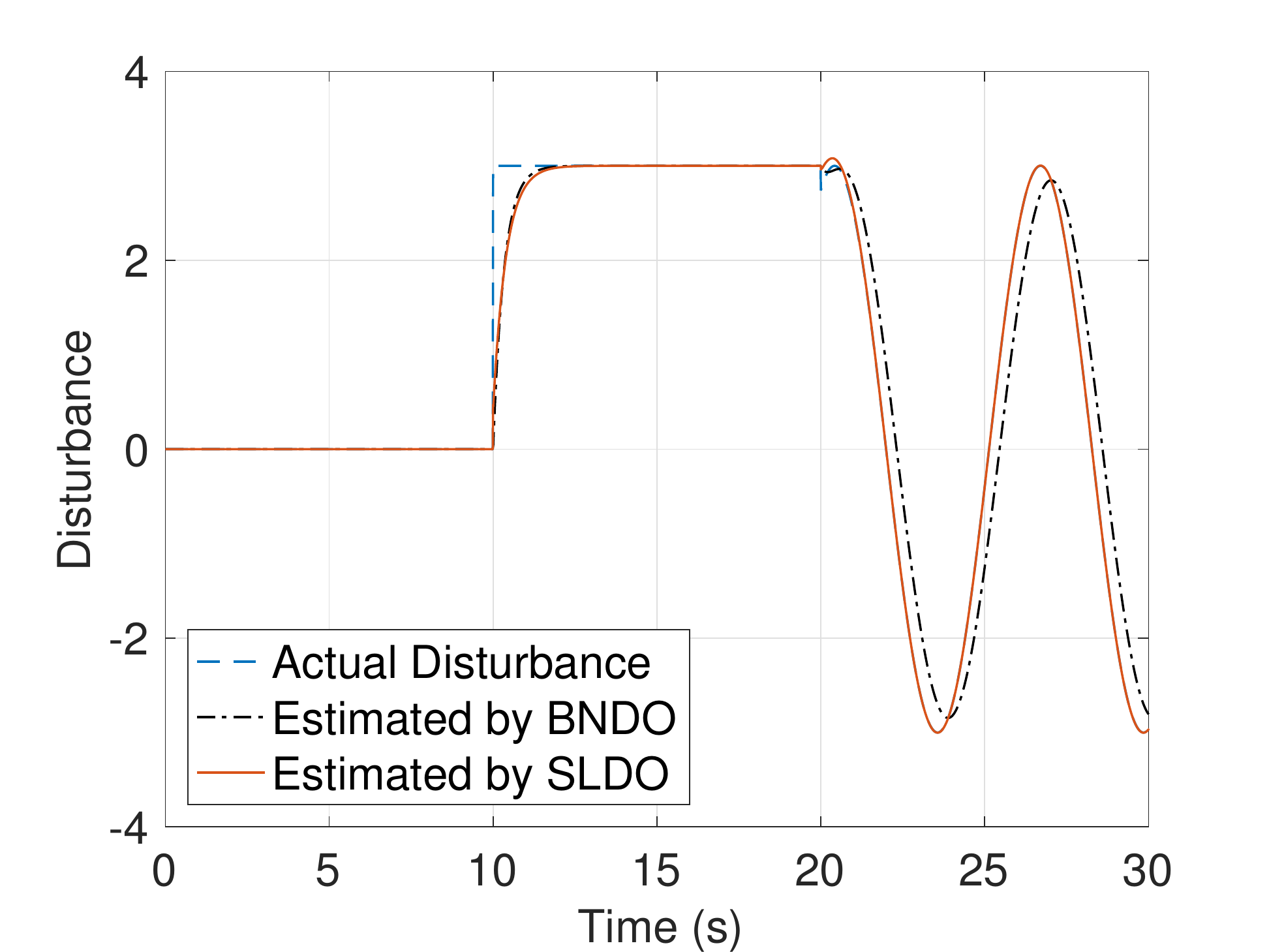}
\label{fig_d}
}
\subfigure[ ]{
\includegraphics[width=3in]{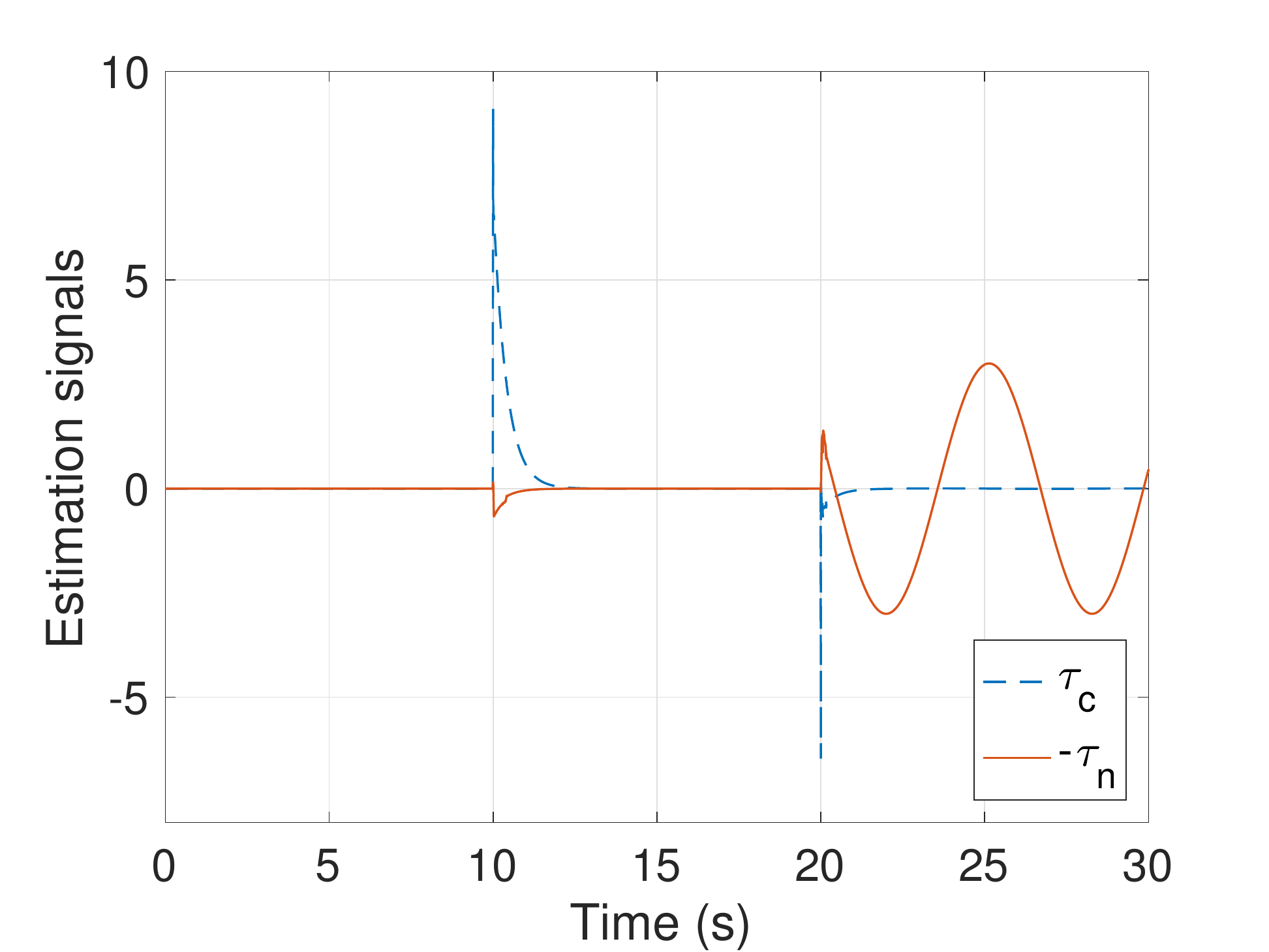}
\label{fig_estimation_signals}
}
\subfigure[ ]{
\includegraphics[width=3in]{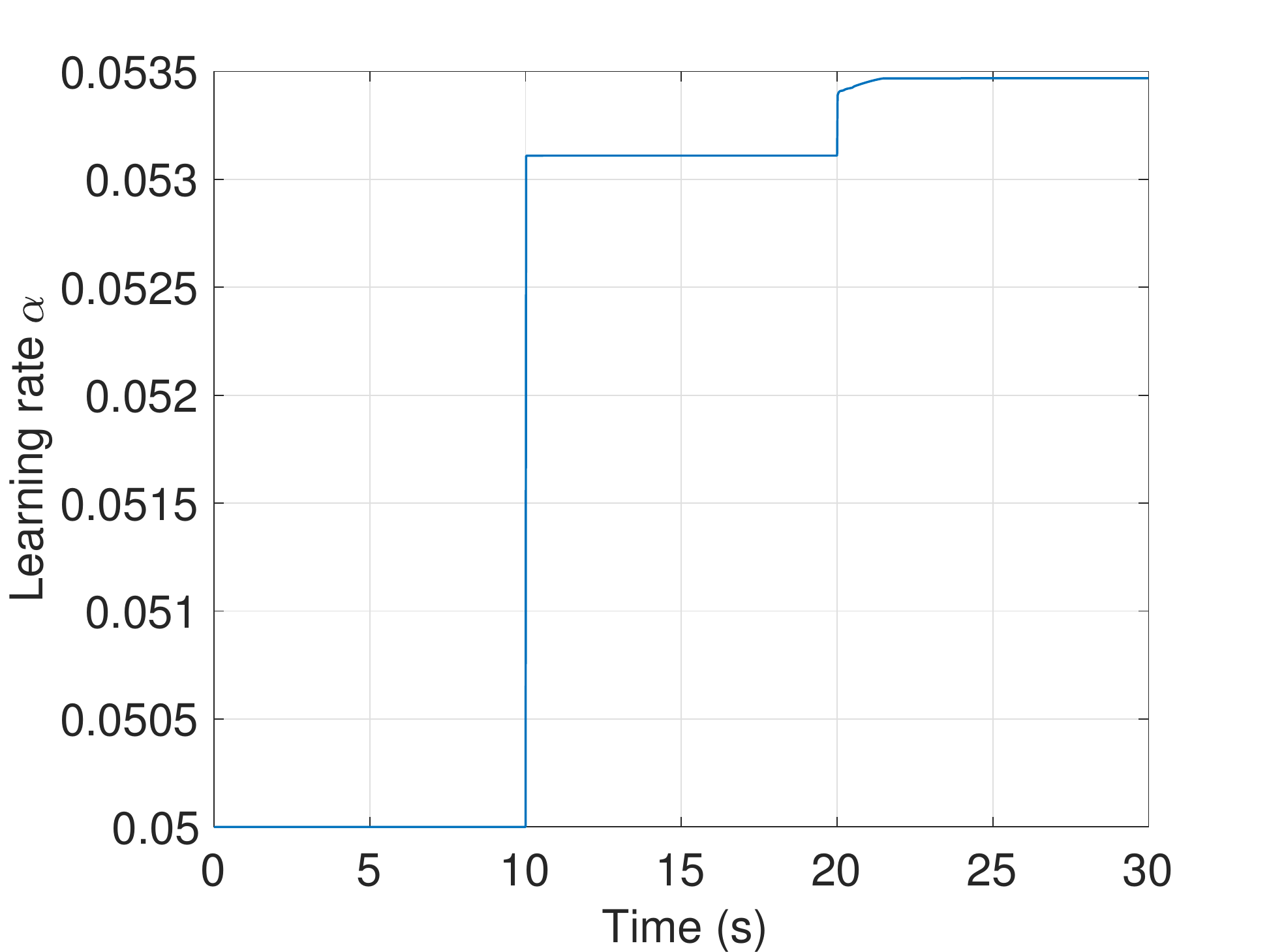}
\label{fig_alpha}
}
\subfigure[ ]{
\includegraphics[width=3in]{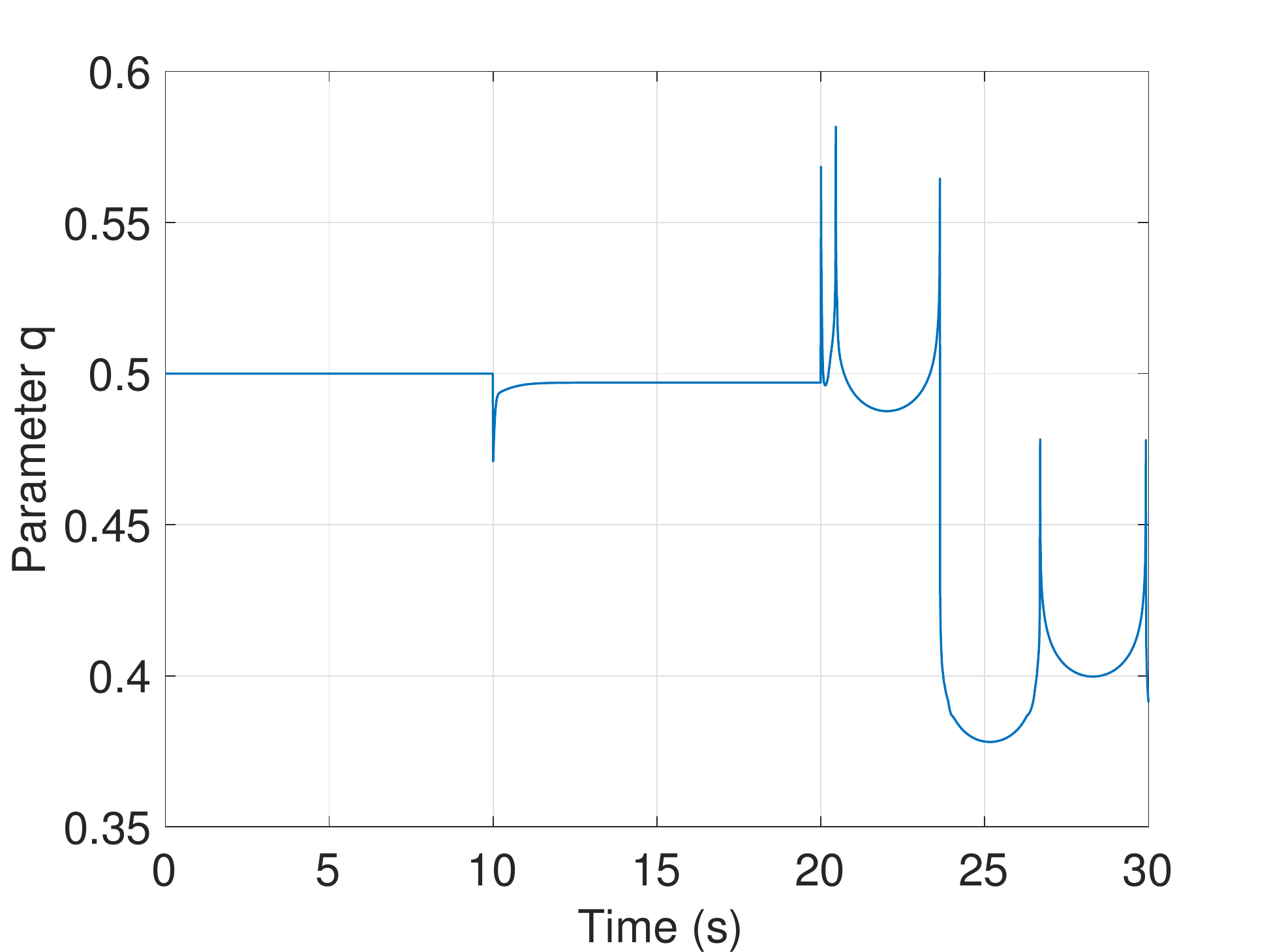}
\label{fig_q}

}
\caption[Optional caption for list of figures]{ (a) Responses of state $x_{1}$(b) Responses of state $x_{2}$ (c) True and estimated values of the disturbance (d) Estimation signals (e) Learning parameter $\alpha$ (f) Parameter $q$}
\label{sensors}
\end{figure}

Type-2 fuzzy membership functions are used in the proposed estimation structure and it is possible to downgrade them to type-1 counterparts by equalizing the upper and lower values of parameters in \eqref{eq_c_1i_lower}-\eqref{eq_sigma_2j_upper}. In literature, it is claimed that the type-2 fuzzy logic system gives better performance than its type-1 counterpart in the presence of noise and uncertainty in the system. The initial conditions on the states of the system are set to $x(0)=[1, 1]^{T}$. In order to compare the performance of T2NFS with its type-1 counterpart under noisy conditions, the actual disturbance $d$ with different noise levels $SNR$, which are equal to $20$ $dB$, $40$ $dB$ and $60$ $dB$, is applied the system. The mean squared errors for the different noise levels are given in Table \ref{tab_mse}. As seen from this table, the T2NFS gives less error than the type-1 neuro-fuzzy structure (T1NFS) and the performance of T2NFS is more remarkable while the noise level is increasing.
Figure \ref{fig_absdiserror} shows the absolute disturbance error responses with a noise level of $20$ $dB$ for T1NFS and T2NFS. As can be observed, T2NFS gives less disturbance error when compared to its type-1 counterpart. Moreover, the system states under noisy condition are shown in Figs. \ref{fig_x1noise} and \ref{fig_x2noise}. The FLC with the SLDO based on T2NFS exhibits better control performance. These results verify previous results seen in \citep{Mendel2000913,Hsiao20081696,5772027}. Type-2 fuzzy systems have more degrees of freedom so that they have capability of dealing with noisy measurements and uncertainties in the system more effectively.

\begin{figure}[h!]
\centering
\subfigure[ ]{
\includegraphics[width=3in]{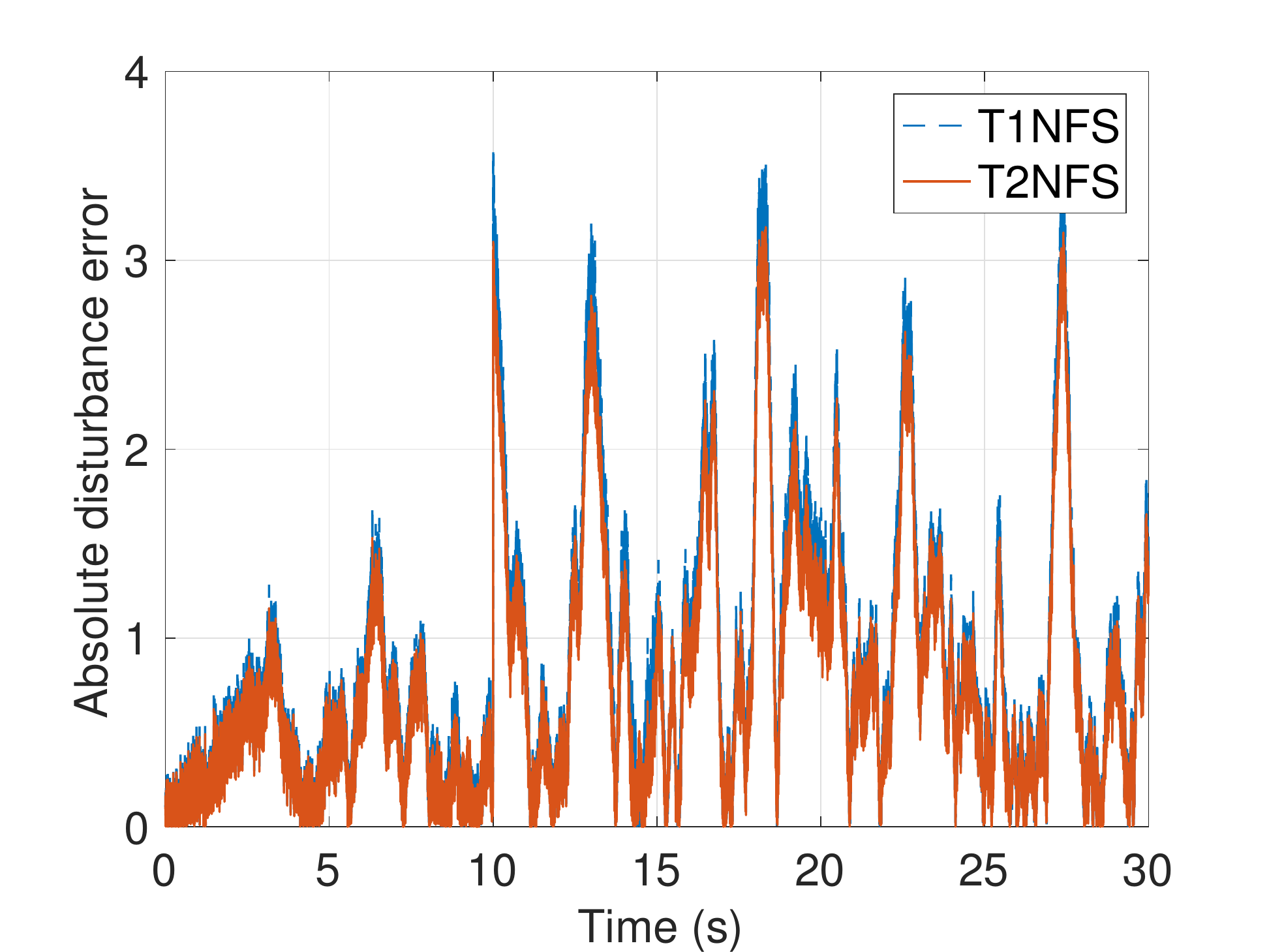}
\label{fig_absdiserror}
}
\subfigure[ ]{
\includegraphics[width=3in]{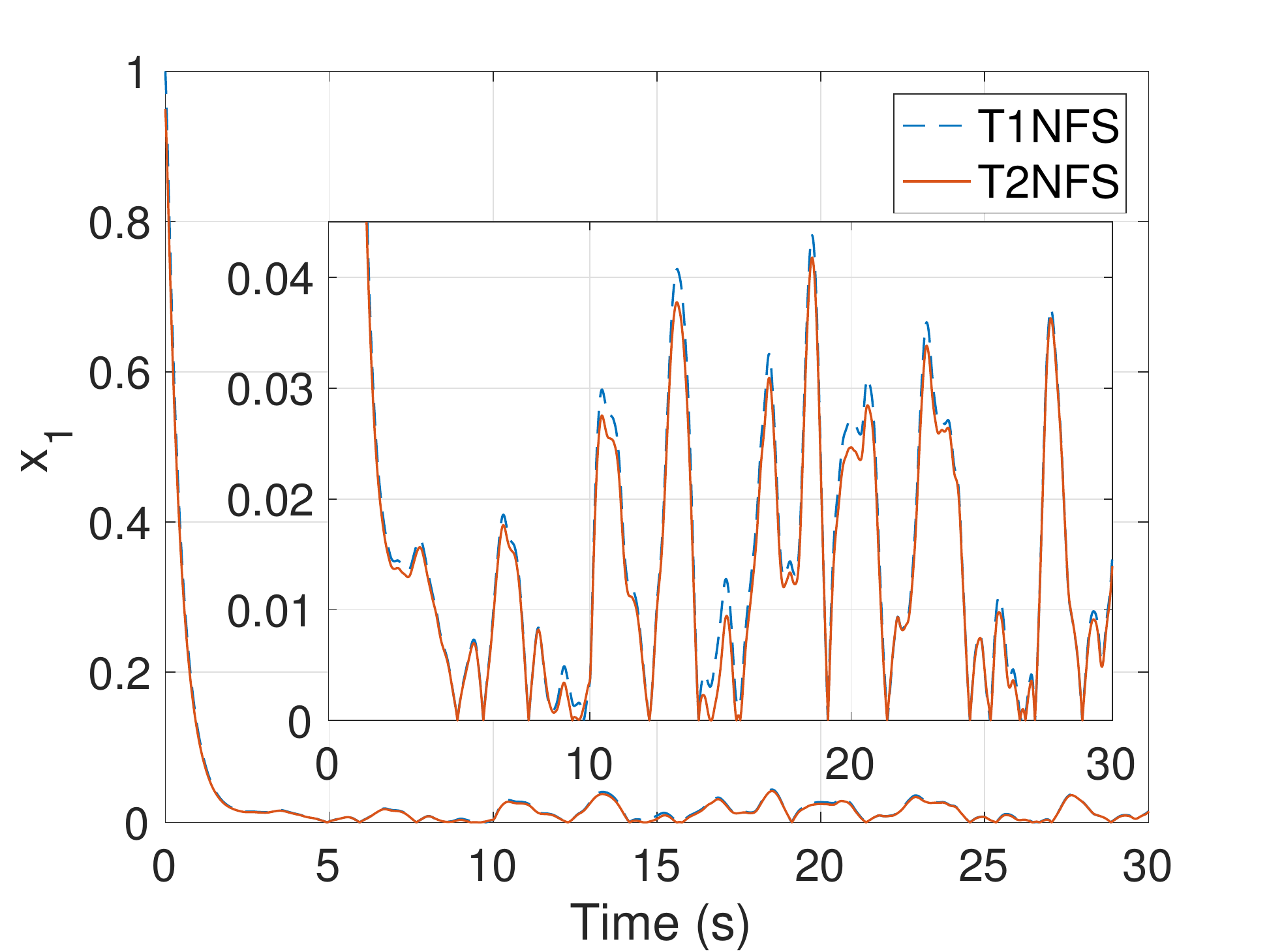}
\label{fig_x1noise}
}
\subfigure[ ]{
\includegraphics[width=3in]{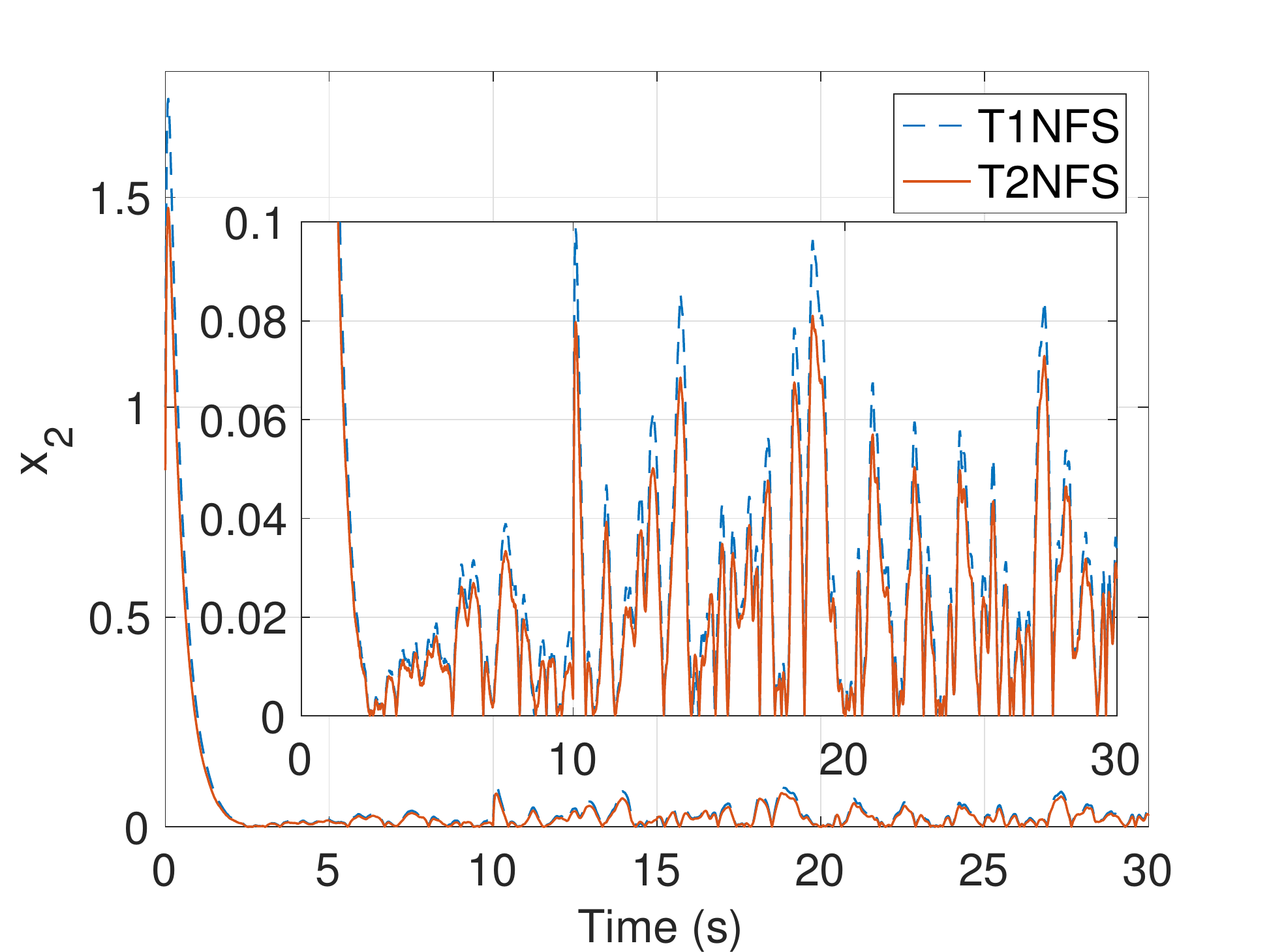}
\label{fig_x2noise}
}
\caption[Optional caption for list of figures]{ (a) Disturbance error when the actual disturbance with a noise level $SNR=20$ $db$  (b) Responses of state $x_{1}$ when the actual disturbance with a noise level $SNR=20$ $db$ (c) Responses of state $x_{2}$ when the actual disturbance with a noise level $SNR=20$ $db$ }
\label{sensors}
\end{figure}

\begin{table}[h!]
\centering
\caption{Mean Squared Error.}\label{tab_mse}
\begin{tabular}{lccc}
  \hline
   & 20 (dB)& 40 (dB)  & 60 (dB)  \\
    \hline
      \hline
  T1NFS & 1.3112 & 0.0288 & 0.0034 \\
  T2NFS & 1.2352 & 0.0276 & 0.0033  \\
  \hline
    \hline
  Decrease in Disturbance Error & 6.14 \% & 4.35 \% & 3.03 \% \\
      \hline
\end{tabular}
\end{table}

The simulations are executed on the computer, which is equipped with $3.1$ GHz Intel Core i$7-5557U$ CPU and $16$ GB of RAM. It is to be noted that the sampling and simulation times are set to $0.001$ and $30$ seconds in the sense that total number of sample is equal to $30000$. The total required computation time is calculated as $5.531993$ seconds so that the required computation time for each sampling time instant is around $0.18$ millisecond. 
The required computation time of the sliding mode learning algorithm is significantly lower than the one of the other methods, such as gradient descent, Levenberg-Marquardt, particle swarm optimization and extended Kalman filter \citep{erkan2015identt2fnn}. The reason is that the sliding mode learning algorithm does not contain any high-order matrices, matrix manipulations or calculations of the partial derivatives. Moreover, particle filter becomes infeasible in real-time due to large number of states and moving horizon estimation requires solving nonlinear optimization problem which results in large computation times \citep{Daum2005}. Recent developed numerical methods have reduced required computation times for solving nonlinear optimization problems around 5 milliseconds. However, even this is 25 times more than our proposed algorithms in the paper \citep{erkancenmpc}. Therefore, it can be concluded that the proposed method in this paper is more practical in real-time applications.

\section{Conclusion}\label{sec_conc}

A novel SLDO is developed by benefiting from the T2NFS with online sliding mode learning algorithm in feedback-error learning scheme. In addition to the stability of the SMC-theory learning  algorithm, the stability of the SLDO by taking the system dynamics into account and the stability of the FLC based on the SLDO are proven by using separation principle. The simulations results show that the traditional FLC is sensitive to disturbances and the FLC based on the BNDO is only robust to time-invariant disturbance while the FLC based on the SLDO is robust against any kind of disturbances by performing precise online estimation of the immeasurable time-varying disturbances. Moreover, the FLC based on the SLDO maintains the nominal control performance in the absence of uncertainties. Thanks to online sliding mode learning algorithm, the parameters of the T2NFS are spontaneously adjusted to learn disturbances and this makes systems robust to cope with uncertainties. Moreover, the developed SMC-theory based learning algorithm requires significantly less computation time than the traditional ones, e.g. gradient descent and evolutionary training algorithms, so that it is more practical in real-time applications.

\appendix
\section{ Calculation of $\dot{\tau}_{n}$}

By taking the time derivative of \eqref{eq_mu1_lower}-\eqref{eq_mu2_upper}, the following equations are obtained as:
\begin{eqnarray}\label{dot_mu}
\dot{\underline{\mu}}_{1i}(\xi_1) &=& -2 \underline{N}_{1i} \dot{\underline{N}}_{1i} \underline{\mu}_{1i}(\xi_1) \nonumber \\
\dot{\overline{\mu}}_{1i}(\xi_1) &=& -2 \overline{N}_{1i} \dot{\overline{N}}_{1i} \overline{\mu}_{1i}(\xi_1)  \nonumber \\
\dot{\underline{\mu}}_{2j}(\xi_2) &=& -2 \underline{N}_{2j} \dot{\underline{N}}_{2j} \underline{\mu}_{2j}(\xi_2)  \nonumber \\
\dot{\overline{\mu}}_{2j}(\xi_2) &=& -2 \overline{N}_{2j} \dot{\overline{N}}_{2j} \overline{\mu}_{2j}(\xi_2)
\end{eqnarray}
where 
\begin{eqnarray}\label{N_dot_N}
\underline{N}_{1i}=\Big(\frac{\xi_1-\underline{c}_{1i}}{\underline{\sigma}_{1i}}\Big), \quad \dot{\underline{N}}_{1i}= \frac{(\dot{\xi_1} - \dot{\underline{c}}_{{1i}})\underline{\sigma}_{1i}-(\xi_1 - \underline{c}_{1i})\dot{\underline{\sigma}}_{1i}}{\underline{\sigma}^2 _{1i}} \nonumber \\
\underline{N}_{2j}=\Big(\frac{\xi_2-\underline{c}_{2j}}{\underline{\sigma}_{2j}}\Big),   \quad \dot{\underline{N}}_{2j}= \frac{(\dot{\xi_2} - \dot{\underline{c}}_{{2i}})\underline{\sigma}_{2i}-(\xi_2 - \underline{c}_{2i})\dot{\underline{\sigma}}_{2i}}{\underline{\sigma}^2 _{2i}}  \nonumber \\
\overline{N}_{1i}=\Big(\frac{\xi_1-\overline{c}_{1i}}{\overline{\sigma}_{1i}}\Big), \quad \dot{\overline{N}}_{1i}= \frac{(\dot{\xi_1} - \dot{\overline{c}}_{{1i}})\overline{\sigma}_{1i}-(\xi_1 - \overline{c}_{1i})\dot{\overline{\sigma}}_{1i}}{\overline{\sigma}^2 _{1i}}  \nonumber \\
\overline{N}_{2j}=\Big(\frac{\xi_2-\overline{c}_{2j}}{\overline{\sigma}_{2j}}\Big) \quad \dot{\overline{N}}_{2j}= \frac{(\dot{\xi_2} - \dot{\overline{c}}_{{2i}})\overline{\sigma}_{2i}-(\xi_2 - \overline{c}_{2i})\dot{\overline{\sigma}}_{2i}}{\overline{\sigma}^2 _{2i}}
\end{eqnarray}

It is obtained from \eqref{N_dot_N} 
\begin{equation}\label{up8}
\underline{N}_{1i}\dot{\underline{N}}_{1i}=\underline{N}_{2j}\dot{\underline{N}}_{2j}=\overline{N}_{1i}\dot{\overline{N}}_{1i}=\overline{N}_{2j}\dot{\overline{N}}_{2j}= \alpha sgn{(s)}
\end{equation}

By taking the time derivative of \eqref{eq_wij_lower_upper_normalized}, the following equations are obtained as follows:
\begin{eqnarray}
\dot{\widetilde{\underline{w}}}_{ij}  = \frac{\Big(\underline{\mu}_{1i}(\xi_1) \underline{\mu}_{2j}(\xi_2)\Big)'\Big(\sum_{i=1}^{I}\sum_{j=1}^{J}\underline{w}_{ij}\Big) - \underline{w}_{ij}\Big(\sum_{i=1}^{I}\sum_{j=1}^{J}\underline{\mu}_{1i}(\xi_1)   \underline{\mu}_{2j}(\xi_2)\Big)'}{(\sum_{i=1}^{I}\sum_{j=1}^{J}\underline{w}_{ij})^2} 
\end{eqnarray}

Since $\widetilde{\underline{w}}_{ij} = \frac{\underline{w}_{ij}}{\sum_{i=1}^{I}\sum_{j=1}^{J}\underline{w}_{ij}}$,
\begin{eqnarray}\label{dotwij_lower_normalized}
\dot{\widetilde{\underline{w}}}_{ij} & = & \frac{\Big(\dot{\underline{\mu}}_{1i}(\xi_1) \underline{\mu}_{2j}(\xi_2)+ \underline{\mu}_{1i}(\xi_1) \dot{\underline{\mu}}_{2j}(\xi_2)\Big)}{\sum_{i=1}^{I}\sum_{j=1}^{J}\underline{w}_{ij}} \nonumber \\
&& -\frac{\widetilde{\underline{w}}_{ij} \bigg(\sum_{i=1}^{I}\sum_{j=1}^{J} \Big(\dot{\underline{\mu}}_{1i}(\xi_1) \underline{\mu}_{2j}(\xi_2)+ \underline{\mu}_{1i}(\xi_1) \dot{\underline{\mu}}_{2j}(\xi_2)\Big)\bigg)}{\sum_{i=1}^{I}\sum_{j=1}^{J}\underline{w}_{ij}} \nonumber \\
& = &  \frac{ -2 \underline{\mu}_{1i}(\xi_1) \underline{\mu}_{2j}(\xi_2) + 2 \widetilde{\underline{w}}_{ij} \sum_{i=1}^{I}\sum_{j=1}^{J}  \underline{\mu}_{1i}(\xi_1)\underline{\mu}_{2j}(\xi_2)}{\sum_{i=1}^{I}\sum_{j=1}^{J}\underline{w}_{ij}}  \Big(\underline{N}_{1i} \dot{\underline{N}}_{1i} + \underline{N}_{2j} \dot{\underline{N}}_{2j} \Big)
\nonumber \\
&=& -\widetilde{\underline{w}}_{ij} \dot{\underline{K}}_{ij} + \widetilde{\underline{w}}_{ij} \sum_{i=1}^{I}\sum_{j=1}^{J}\widetilde{\underline{w}}_{ij} \dot{\underline{K}}_{ij}
\end{eqnarray}
where
\begin{equation*}
\dot{\underline{K}}_{ij} = 2 \Big(\underline{N}_{1i} \dot{\underline{N}}_{1i}+\underline{N}_{2j}\dot{\underline{N}}_{2j} \Big) = 4 \alpha sgn{(s)}
\end{equation*}

Similarly, it is readily obtained that:
\begin{equation}\label{dotwij_upper_normalized}
\dot{\widetilde{\overline{w}}}_{ij} = -\widetilde{\overline{w}}_{ij}\dot{\overline{K}}_{ij} + \widetilde{\overline{w}}_{ij} \sum_{i=1}^{I}\sum_{j=1}^{J}\widetilde{\overline{w}}_{ij} \dot{\overline{K}}_{ij}
\end{equation}
where
\begin{equation*}
\dot{\overline{K}}_{ij} = 2 \Big(\overline{N}_{1i}\dot{\overline{N}}_{1i}+\overline{N}_{2j}\dot{\overline{N}}_{2j}\Big) = 4 \alpha sgn{(s)}
\end{equation*}

The time derivative of \eqref{eq_taun} is obtained to find $\dot{\tau}_n$ as follows:
\begin{eqnarray}
\dot{\tau}_n & = & q\sum_{i=1}^{I}\sum_{j=1}^{J}(\dot{f}_{ij} \widetilde{\underline{w}}_{ij}+ f_{ij}\dot{\widetilde{\underline{w}}}_{ij}) +(1-q)\sum_{i=1}^{I}\sum_{j=1}^{J}(\dot{f}_{ij}\widetilde{\overline{w}}_{ij} + f_{ij}\dot{\widetilde{\overline{w}}}_{ij})  \nonumber \\
& & +\dot{q} \sum_{i=1}^{I}\sum_{j=1}^{J}f_{ij}\widetilde{\underline{w}}_{ij}
-\dot{q} \sum_{i=1}^{I}\sum_{j=1}^{J}f_{ij}\widetilde{\overline{w}}_{ij}
\end{eqnarray}

If \eqref{dotwij_lower_normalized} and \eqref{dotwij_upper_normalized} are inserted into the aforementioned equation:
\begin{eqnarray}\label{dotVc2}
\dot{\tau}_n & = & q\sum_{i=1}^{I}\sum_{j=1}^{J}\bigg(\Big(-\widetilde{\underline{w}}_{ij} \dot{\underline{K}}_{ij} +\widetilde{\underline{w}}_{ij} \sum_{i=1}^{I}\sum_{j=1}^{J}\widetilde{\underline{w}}_{ij} \dot{\underline{K}}_{ij} \Big)f_{ij} +\widetilde{\underline{w}}_{ij}\dot{f}_{ij}\bigg) \nonumber \\
&&+(1-q)\sum_{i=1}^{I}\sum_{j=1}^{J}\bigg(\Big( -\widetilde{\overline{w}}_{ij} \dot{\overline{K}}_{ij} + \widetilde{\overline{w}}_{ij} \sum_{i=1}^{I}\sum_{j=1}^{J}\widetilde{\overline{w}}_{ij} \dot{\overline{K}}_{ij} \Big)f_{ij} +\widetilde{\overline{w}}_{ij}\dot{f}_{ij}\bigg) \nonumber \\
&& +\dot{q} \sum_{i=1}^{I}\sum_{j=1}^{J}f_{ij} (\widetilde{\underline{w}}_{ij} - \widetilde{\overline{w}}_{ij} )   \nonumber \\
& = & 4 \alpha sgn{(s)} \Bigg(q \sum_{i=1}^{I}\sum_{j=1}^{J}\bigg(\Big(-\widetilde{\underline{w}}_{ij}   + \widetilde{\underline{w}}_{ij} \sum_{i=1}^{I}\sum_{j=1}^{J}\widetilde{\underline{w}}_{ij}  \Big)f_{ij}  +\widetilde{\underline{w}}_{ij}\dot{f}_{ij}\bigg) \nonumber \\
&& +(1-q)\sum_{i=1}^{I}\sum_{j=1}^{J}\bigg(\Big( -\widetilde{\overline{w}}_{ij}  + \widetilde{\overline{w}}_{ij} \sum_{i=1}^{I}\sum_{j=1}^{J}\widetilde{\overline{w}}_{ij} \Big)f_{ij} +\widetilde{\overline{w}}_{ij}\dot{f}_{ij}\bigg) \Bigg) \nonumber \\
&& +\dot{q} \sum_{i=1}^{I}\sum_{j=1}^{J}f_{ij} (\widetilde{\underline{w}}_{ij} - \widetilde{\overline{w}}_{ij} ) 
\end{eqnarray}

Since $\sum_{i=1}^{I}\sum_{j=1}^{J}\widetilde{\overline{w}}_{ij} =1$ and $\sum_{i=1}^{I}\sum_{j=1}^{J}\widetilde{\underline{w}}_{ij} =1$, the aforementioned equation becomes by using \eqref{eq_f_ij} and \eqref{eq_q} as follows:
\begin{eqnarray}\label{dotVc4}
\dot{\tau}_n & = & q \sum_{i=1}^{I}\sum_{j=1}^{J}\widetilde{\underline{w}}_{ij} \dot{f}_{ij} + (1-q)\sum_{i=1}^{I}\sum_{j=1}^{J}\widetilde{\overline{w}}_{ij} \dot{f}_{ij}  +\dot{q} \sum_{i=1}^{I}\sum_{j=1}^{J}f_{ij} (\widetilde{\underline{w}}_{ij} - \widetilde{\overline{w}}_{ij} )  \nonumber \\
& = &- 2 \alpha sgn{(s)}
\end{eqnarray}

\section*{Acknowledgement}
The information, data, or work presented herein was funded in part by the Advanced Research Projects Agency-Energy (ARPA-E), U.S. Department of Energy, under Award Number DE-AR0000598.

\bibliographystyle{elsarticle-harv}
\bibliography{reference}

\end{document}